\documentclass[conference]{IEEEtran}
\IEEEoverridecommandlockouts

\usepackage[linesnumbered,ruled,vlined]{algorithm2e}
\usepackage{mathrsfs}
\usepackage{amsmath,amsfonts,amsthm}
\usepackage{bm}
\usepackage{colortbl}
\usepackage{cases}
\usepackage{bbding}
\usepackage{pifont}
\usepackage{subfigure}
\usepackage{xcolor}
\usepackage{tabularx}
\usepackage{booktabs}
\usepackage{threeparttable}
\usepackage{graphicx}
\usepackage[export]{adjustbox}
\usepackage{hyperref}
\usepackage{epstopdf}
\usepackage{caption}
\usepackage{ulem}
\usepackage{hhline}
\usepackage[OT1]{fontenc}
\usepackage[noend]{algpseudocode}
\usepackage{longtable}
\usepackage{soul}
\usepackage{wasysym}
\usepackage{float} 
\newtheorem{theorem}{Theorem}

\newtheorem{definition}{Definition}

\definecolor{BrickRed}{RGB}{178,34,34}

\AtBeginDocument{%
  }
    
\makeatletter  
\newif\if@restonecol  
\makeatother

\usepackage[
	n,
	operators,
	advantage,
	sets,
	adversary,
	landau,
	probability,
	notions,	
	ff,
	mm,
	primitives,
	events,
	complexity,
	asymptotics,
	keys]{cryptocode}

\begin{document}


\title{FileDES: A Secure, Scalable and Succinct Decentralized Encrypted Storage Network
\thanks{Corresponding author: Xiuzhen Cheng (\href{mailto:xzcheng@sdu.edu.cn}{xzcheng@sdu.edu.cn}).}
}

\author{
	\IEEEauthorblockN{Minghui~Xu$^{\dag\P}$, Jiahao Zhang$^{\dag}$, Hechuan Guo$^{\dag}$, Xiuzhen~Cheng$^{\dag}$, Dongxiao Yu$^{\dag}$, Qin Hu$^{\ddagger}$, Yijun Li$^{*}$, Yipu Wu$^{*}$}
	\IEEEauthorblockA{$^\dag$ School of Computer Science and Technology, Shandong University}
        \IEEEauthorblockA{$^\P$ ETH Zurich}
        \IEEEauthorblockA{$^\ddagger$ Department of Computer and Information Science, Indiana University-Purdue University Indianapolis}
        \IEEEauthorblockA{$^*$ BaishanCloud}
}

\maketitle

\begin{abstract}
Decentralized Storage Network (DSN) is an emerging technology that challenges traditional cloud-based storage systems by consolidating storage capacities from independent providers and coordinating to provide decentralized storage and retrieval services. However, current DSNs face several challenges associated with data privacy and efficiency of the proof systems. To address these issues, we propose FileDES (\uline{D}ecentralized \uline{E}ncrypted \uline{S}torage), which incorporates three essential elements: privacy preservation, scalable storage proof, and batch verification. FileDES provides encrypted data storage while maintaining data availability, with a scalable Proof of Encrypted Storage (PoES) algorithm that is resilient to Sybil and Generation attacks. Additionally, we introduce a rollup-based batch verification approach to simultaneously verify multiple files using publicly verifiable succinct proofs. We conducted a comparative evaluation on FileDES, Filecoin, Storj and Sia under various conditions, including a WAN composed of up to 120 geographically dispersed nodes. Our protocol outperforms the others in terms of proof generation/verification efficiency, storage costs, and scalability.
\end{abstract}

\begin{IEEEkeywords}
Decentralized storage network, blockchain, data sharing, proof of storage, scalability, Sybil attacks. 
\end{IEEEkeywords}

\IEEEpeerreviewmaketitle

\section{Introduction}
\label{sec:introduction}
Blockchain technology has brought about a significant innovation in distributed storage. Decentralized storage networks (DSNs) represent a novel approach that can aggregate available storage spaces from independent providers, allowing for coordinated and reliable storage and retrieval of data. Several well-known DSN projects, including Filecoin~\cite{filecoin}, Sia~\cite{sia}, Storj~\cite{storj}, and Swarm~\cite{swarm}, have demonstrated various advantages of DSNs, including storage capacity expansion, data sharing promotion, and data security enhancement. By incentivizing storage providers with cryptocurrency rewards, DSNs can achieve larger capacity than traditional distributed storage networks~\cite{kopp2017design}. The use of blockchain ensures consistency and immutability, improving the security and robustness of services provided by mutually untrusted storage providers. DSNs have proven to be a valuable building block for applications such as Web 3.0~\cite{korpal2022decentralization}. However, three open challenges \textbf{[C1-C3]}, as illustrated in Fig.~\ref{fig:challenges}, still exist, which significantly impede the performance and security of today's DSNs.

\begin{figure}[!htbp] 
	\centering 
	\includegraphics[width=0.45\textwidth]{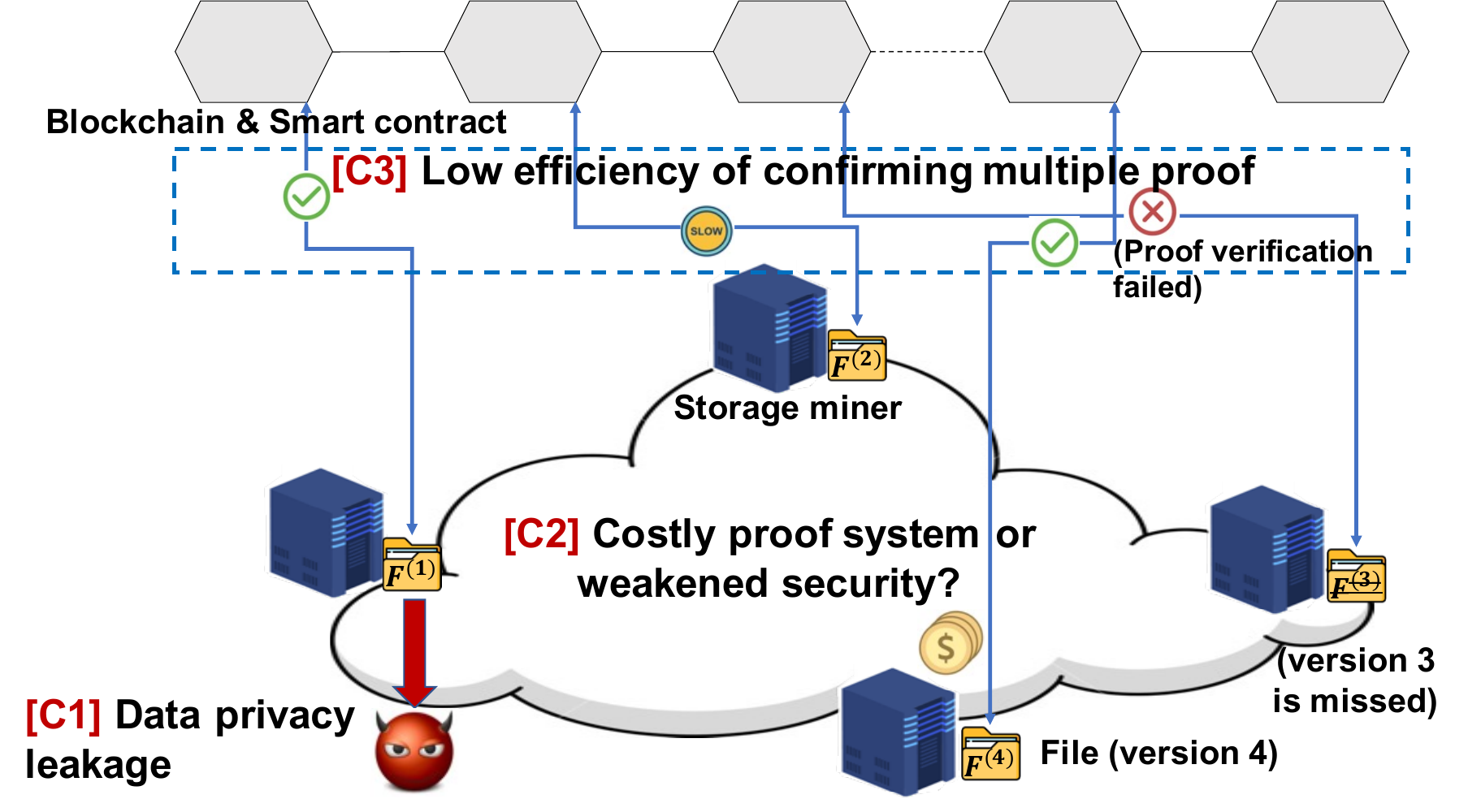} 
	\caption{Open challenges faced by today's DSNs.} 
	\label{fig:challenges} 
\end{figure}

\noindent \textbf{[C1] Data privacy leakage.} 
A fundamental goal of DSNs is to promote data sharing. Nevertheless, current solutions present two primary options for data storage, namely, plaintext and simple encrypted data. The former approach, which has been widely adopted, permits the direct storage of non-sensitive plaintext by users~\cite{psaras2020interplanetary}. However, it falls short of adequate safeguards for data privacy. The latter option, exemplified by the ChainSafe Files~\cite{chainsafe}, allows for the storage of encrypted data, which are exclusively visible only to designated users. As a result, this approach adversely affects data availability, thereby hindering data sharing. Thus, there is a pressing need to develop a DSN solution that can guarantee robust data privacy and high data availability.

\noindent \textbf{[C2] Costly or security-weakened proof system.} 
Proof of Space (PoS) and Proof of Spacetime (PoSt) are essential components of the proof system of a DSN. These protocols enable the network to validate whether a storage miner has offered legitimate storage services by issuing real-time challenges and demanding responses as the corresponding proofs~\cite{benet2017proof}. Nonetheless, generating a proof can be a time-intensive and hardware-demanding process~\cite{fisch2019tight, fisch2018poreps, ren2016proof, guo2022filedag}, or require the compromise of security for efficiency. In the Filecoin network, for instance, the proof generation (i.e., the $\mathtt{SEALING}$ process) involves a complex data structure called a Stacked Depth Robust Graph, making it computationally expensive. Thus, a minimum hardware configuration of 256GB RAM and a GPU with 11GB VRAM is required to participate in storage services~\footnote{https://lotus.filecoin.io/storage-providers/get-started/hardware-requirements/}. Other DSNs, such as Sia~\cite{sia}, Storj~\cite{storj}, and Swarm~\cite{swarm}, have simplified their proof systems to enhance efficiency at the expense of sacrificing security. Specifically, Sia and Swarm only provide proofs for partial data (256KB sector in Sia and 1MB stipe in Swarm), while Storj's proof system is based on node reputation, leading to non-negligible false positives~\cite{storj}. Overall, the efficiency and security of DSNs' proof systems are crucial to ensuring the system' viability and adoption.

\noindent \textbf{[C3] Low efficiency of recurrent proof verification.} 
In DSNs, one of the recurrent tasks is verifying the PoS and PoSt. This verification process can pose significant computational challenges, especially in the following two situations. (1) The challenger frequently challenges the storage miners to verify that they have not violated any rule in storing files. For example, the Filecoin mainnet challenges a file every 30 minutes, 48 times daily~\cite{wursten2022filecoin}, to verify enormous proofs. (2) Managing multiple versions of a file can be a complex task, requiring the verification of several versions simultaneously. This is essential to ensure that historical files of a project are not lost. Such situations are prevalent, particularly in software updates~\cite{nikitin2017chainiac} and medical records maintenance~\cite{dubovitskaya2017secure}. To expedite the frequent verification process in these applications, an efficient proof system is indispensable.
 
In response, we propose FileDES to address these challenges. Our major contributions can be summarized as follows:
\begin{itemize}
    \item FileDES is a decentralized encrypted storage network that addresses the challenges [C1-C3]. It offers unique features such as privacy preservation, scalable storage proof, and batch proof verification. In comparison with Filecoin and Sia, FileDES demonstrates superior performance under various typical conditions.
    \item To protect data privacy while maintaining data availability, we introduce an encrypted storage plan based on RSA and Unidirectional Proxy Re-Encryption (PRE). This method not only prevents Sybil and Generation attacks but also facilitates a secure storage mechanism without complicated PoS \& PoSt schemes.
    \item We propose a new Proof of Encrypted Storage (PoES) that provides unforgeable PoS \& PoSt on encrypted data in an efficient and scalable manner. To mitigate Sybil attacks, we incorporate a random storage miner selection algorithm as an incentive mechanism.
    \item To reduce storage redundancy, FileDES stores only file increments. Additionally, we introduce a rollup-based batch verification method that verifies multiple proofs with only one publicly verifiable proof submitted to the blockchain.
\end{itemize}

\section{Related Work}
\label{sec:related:work}

\begin{table*}[!t]
	\begin{center}
	\begin{threeparttable}
		\caption{Comparison of Our Work with Existing DSNs}
            \label{table:features}
            \tabcolsep=0.2cm
			\begin{tabular}{l c c c c c c c c c}
				\toprule[1pt]
				 & \multicolumn{1}{c}{\textbf{\begin{tabular}[c]{@{}c@{}}Consensus\\ Algorithm\end{tabular}}} 
				 & \textbf{Ledger} 
				 & \multicolumn{1}{c}{\textbf{\begin{tabular}[c]{@{}c@{}}Security\\ Level\end{tabular}}}
				 & \multicolumn{1}{c}{\textbf{\begin{tabular}[c]{@{}c@{}}Low\\ Redundancy\end{tabular}}} 
				 & \multicolumn{1}{c}{\textbf{\begin{tabular}[c]{@{}c@{}}Privacy \\ Protection\end{tabular}}} 
                & \multicolumn{1}{c}{\textbf{\begin{tabular}[c]{@{}c@{}}Data Availability \\ after Encryption\end{tabular}}} 
				 & \multicolumn{1}{c}{\textbf{\begin{tabular}[c]{@{}c@{}}Storage proof\\ method\end{tabular}}} 
				 & \textbf{Scalability}
                & \multicolumn{1}{c}{\textbf{\begin{tabular}[c]{@{}c@{}}Batch\\ Verification\end{tabular}}}\\
                \midrule[0.5pt]
				
				Filecoin\cite{filecoin} & \multicolumn{1}{c}{\begin{tabular}[c]{@{}c@{}}Expected\\ consensus\end{tabular}} & \multicolumn{1}{c}{\begin{tabular}[c]{@{}c@{}}DAG\\ (tipset)\end{tabular}} & \CIRCLE & \Circle & \LEFTcircle & \Circle & DRG+MT & \Circle & \Circle\\

    			FileDAG\cite{guo2022filedag} & DAG-Rider$^\dag$ & DAG & \CIRCLE & \CIRCLE & \LEFTcircle & \Circle & DRG+MT & \LEFTcircle & \Circle\\
        
				Storj\cite{storj} & PoW & Chain & \LEFTcircle & \LEFTcircle & \LEFTcircle & \Circle & Reputation  & \Circle & \Circle\\
				
				Sia\cite{sia} & PoW & Chain & \LEFTcircle & \Circle & \LEFTcircle & \Circle & MT & \LEFTcircle & \Circle\\
				
				Swarm\cite{swarm} & PoW & Chain & \LEFTcircle & \Circle & \LEFTcircle & \Circle & MT & \LEFTcircle & \Circle\\

				\textbf{FileDES} & DAG-Rider$^\dag$ & DAG & \CIRCLE & \CIRCLE & \CIRCLE & \CIRCLE & Encryption+MT & \CIRCLE & \CIRCLE\\
    
				\bottomrule[1pt]
			\end{tabular}
			\label{tab1}
		\begin{tablenotes}
		    \item[DRG] Depth Robust Graph
                \item[MT] Merkle Tree
		    \item[$\dag$] Modified version
		\end{tablenotes}
	\end{threeparttable}
	\end{center}
\end{table*}

\subsection{Decentralized Storage Networks}
Currently, there exist many studies and implementations that aim to develop a DSN. Filecoin~\cite{filecoin} is a DSN that utilizes InterPlanetary File System~\cite{ipfs} (IPFS) and Expected Consensus mechanism to adjust storage miner's chances of winning based on their storage quantity and quality. Filecoin introduces the concept of tipsets to enable concurrent block processing, allowing multiple blocks to be confirmed at the same block height. Sia~\cite{sia} is a DSN that employs the Proof-of-Work consensus and creates a Merkle tree for each file, with the root hash serving as the file identifier. Sia can verify whether a file has been uploaded before and enable deduplication at the directory level. It uses the Threefish~\cite{threefish} algorithm to encrypt files, making it difficult to support version indexing and sharing. Additionally, its ledger structure is a chain, which cannot support concurrent block processing by nature. Storj~\cite{storj} and Swarm~\cite{swarm} are DSNs built on Ethereum~\cite{ethereum}. They both make use of the Proof-of-Stake consensus and a chain-based ledger. 
FileDAG~\cite{guo2022filedag} is a DSN that builds on the implementation of Filecoin. It achieves file-level deduplication when storing multi-version files and employs the DAG-Rider~\cite{dagrider} consensus mechanism to create a two-layer DAG-based blockchain ledger for flexible and storage-saving file indexing. Zhang et al.~\cite{zhang2022enabling} proposed a secure mechanism over decentralized storage by employing smart contracts to incorporate the message-locked encryption~\cite{bellare2013message} (MLE) scheme, which protects data privacy and enables secure deduplication over encrypted data. Ismail et al.~\cite{ismail2022cost} evaluated the costs and latency performance of nine state-of-the-art systems and discussed their compatibility with the decentralized features of blockchain technology.

It is pertinent to note that the concerns raised in [C1-C3] have not been satisfactorily addressed by the current DSNs. In order to substantiate this claim, Table~\ref{table:features} presents a concise comparison study over the attributes of FileDES and those of other popular DSNs. In summary, Sia, Storj, and Swarm have simplified their proof systems to improve efficiency, but at the expense of security. Storj and FileDAG respectively utilize erasure coding and incremental generation to reduce redundancy. For privacy protection, the only approach that can be adopted by these DSNs (except FileDES) is to apply simple encryption but such an approach fails to guarantee data availability for encrypted data. Sia and Swarm improve the storage proof efficiency to enhance the scalability. FileDAG utilizes a DAG-based blockchain to enhance the scalability of the consensus module. Additionally, to our best knowledge, no existing system considers batch verification.

\subsection{Blockchain-based Data Sharing}
The use of blockchain for data sharing is akin to the traditional DSN model concerning the management of user data in a blockchain-based system. Several viable solutions exist to accomplish blockchain-based data sharing.
For example, MedChain~\cite{shen2019medchain} is a healthcare data-sharing scheme that employs blockchain, digest chain, and structured P2P network techniques to enhance efficiency and security in sharing healthcare data. 
SPDL~\cite{9761745} is a decentralized learning system that employs blockchain and Byzantine Fault-Tolerant consensus to facilitate secure and private data sharing during the machine learning process.
Ghostor~\cite{hu2020ghostor} is a data-sharing system that utilizes decentralized trust to safeguard user privacy and data integrity against compromised servers. The system conceals user identities from the server and enables users to detect server-side integrity violations. 
TEMS~\cite{tems} is a framework that extends blockchain trust from on-chain to the physical world by employing a Trusted Execution Environment (TEE) system with anti-forgery data and a consistency protocol to continuously and truthfully upload data to blockchain.

\section{Models and Preliminaries}
\label{sec:model}

\subsection{The DSN Model}
The DSN under our study comprises of four distinct entities, namely client, storage miner, retrieval miner, and rollup miner. A client serves as a media for users to interact with the storage system. A DSN design must incorporate fundamental functionalities, i.e., $\mathsf{Put}$, $\mathsf{Get}$, and $\mathsf{Manage}$,  that can be executed by both clients and miners.
    $$\mathsf{DSN = (Put, Get, Manage}).$$
    \begin{itemize}
	\item $\mathsf{Put(F, SM)}\rightarrow \mathsf{CID}$: A client executes the $\mathsf{Put}$ protocol to upload the file $\mathsf{F}$ to a storage miner $\mathsf{SM}$ for storing the file in the DSN, and obtain a file identifier $\mathsf{CID}$.
	\item $\mathsf{Get(CID, ReM)}\rightarrow \mathsf{F}$. A client executes the $\mathsf{Get}$ protocol to send the file identifier $\mathsf{CID}$ to a retrieval miner $\mathsf{ReM}$ to retrieve data from the DSN.
	\item $\mathsf{Manage(\mathsf{F}, SM, ReM, RoM)}$. This protocol coordinates the network participants to control the available storage, audit services, and repair potential faults. A rollup miner $\mathsf{RoM}$ is capable of generating aggregated proofs for DSN management. 
    \end{itemize}

\subsection{Adversary Model}
\label{sec:attacks}
FileDES considers the presence of a Byzantine adversary who is restricted to controlling no more than $1/3$ of the total number of nodes. If this particular constraint is violated, achieving a consensus becomes impossible due to the Byzantine General Problem~\cite{castro1999practical}. The adversary has at its disposal a variety of attack forms, including Sybil attacks, Generation attacks, falsifying proofs, and collusion. Among the critical issues associated with such attacks are the vulnerabilities caused by Sybil and Generation attackers. Sybil attackers can manipulate the system by creating multiple Sybil identities to receive various replications of a file, while Generation attackers can manipulate the system with a small seed or a program to re-generate a replica they claim to store. To provide a better understanding of these two attacks, we formally define them as follows.

\begin{definition}[Sybil Attack]
The Sybil Attack is a type of security threat in which an attacker, referred to as $\mathcal{A}_{\mathsf{sybil}}$, creates multiple fake identities, known as Sybil identities, denoted as $\{P_{0}...P_{n}\}$. The attacker claims to have stored $m$ different copies of a file $\mathsf{F}$, but in reality, it only stores $m'<m$ copies. The objective of the Sybil attacker is to successfully forge $m$ valid proofs for the $m$ replicas, which can convince any verifier $\mathcal{V}$ that $\mathsf{F}$ is stored as $m$ independent replications by the attacker.
\end{definition}

\begin{definition}[Generation Attack]
The Generation Attack is a type of security threat in which an attacker, referred to as $\mathcal{A}_{\mathsf{gen}}$, claims to store a replica of a file $\mathsf{F}$, but in reality it does not. The attacker succeeds by generating the replica using a small seed or a program, which is much smaller in storage space compared to that of the replica, every time it needs to produce a proof of the replica.
\end{definition}

Sybil and Generation attacks are two distinct types of security threats. $\mathcal{A}_{\mathsf{sybil}}$ generates a large number of Sybil identities, to make multiple claims of storing different replicas of a file. The attack is profitable when $\mathcal{A}_{\mathsf{sybil}}$ abandons some replicas of the file it claims to store. In contrast, the Generation Attack entails a small seed or a program to generate a replica of a file, which the attacker, $\mathcal{A}_{\mathsf{gen}}$, claims to store. This type of attackers aims to conserve storage space by eliminating redundant data, while being capable of restoring the file at any time using the aforementioned malicious program.

\subsection{Preliminaries}
\subsubsection{Unidirectional Proxy Re-Encryption (PRE)}
A unidirectional Predicate Encryption (PRE) scheme can be formally defined as a set of algorithms $$\mathsf{PRE} = (\mathsf{KeyGen}, \\ \mathsf{ReKeyGen}, \mathsf{Enc}, \mathsf{ReEnc}, \mathsf{Dec}).$$
In FileDES, the PRE scheme facilitates a storage miner to transform a ciphertext encrypted using one key into a new ciphertext with a different key, without requiring access to the plaintext. This functionality enables a data owner to share encrypted data with multiple recipients, eliminating the need for it to re-encrypt the data for each recipient.

\subsubsection{zk-SNARK}
The Zero-Knowledge Succinct Non-Interactive Argument of Knowledge (zk-SNARK) enables the verification of the authenticity of a relation, while keeping confidential information undisclosed. We employ zk-SNARK in the creation of compact, fixed-length PoS \& PoSt. Essentially, a zk-SNARK is a non-interactive argument $$\mathsf{ZK} = (\mathsf{Setup, Prove, Verify})$$ for a relation $R$; it upholds the properties of completeness, soundness, zero-knowledge, and succinctness. Note that $\mathsf{ZK.Setup}$ decides a set of public parameters.

\section{The Design of FileDES}
\label{sec:design}
This section first presents the encryption mechanism as a means to safeguard against Sybil and Generation attacks, which simultaneously ensures data privacy. The primary components of FileDES are $\mathsf{PoES}$ and $\mathsf{Rollup}$, which are introduced in Sections \ref{sec:poes} and \ref{sec:multi:version}, respectively. 

\subsection{Encrypted Storage}
The encrypted storage designed for FileDES provides a dual-purpose solution. Firstly, the encryption serves as a preventative measure to counter Sybil attacks and Generation attacks -- once agreed to store encrypted replicas, a Sybil (or Generation) attacker without knowledge of the secret keys is unable to reproduce the replicas with plaintext (or a seed/malicous program). Secondly, encryption enhances data privacy, which addresses the privacy issues highlighted in [C2] -- encrypted data is unreadable without proper authorization, which helps to protect data privacy and prevent data breaches.

Our design incorporates two encryption options based on either RSA or PRE. The rationale for deploying two encryption tools is to cater to the various data sharing scenarios:  data can be accessed freely with RSA or under permission through authorization with PRE. In each of these methods, a client encrypts a file with multiple secret keys to create different encrypted replicas, which are then uploaded to storage miners. The replicas, as ciphertexts, appear to be random to a Generation attacker, making it challenging to partially or completely delete a replica and re-generate it with a malicious program, unless the attacker can manage to steal the secret key. In addition, the best way for a Sybil attacker to succeed in breaking FileDES is to generate all encrypted replicas from the one piece of plaintext. However, even if an attacker knows the plaintext (in RSA-based option) of a replica, it cannot restore all other replicas without the corresponding secret keys. Apart from encryption, FileDES also makes use of a random selection algorithm to mitigate Sybil attacks leveraging a fairness guarantee. A detailed and formal security analysis of FileDES considering Generation and Sybil attacks is elaborated in Section~\ref{sec:analysis}.

\subsection{Proof of Encrypted Storage (PoES)}
\label{sec:poes}
To enhance data privacy and improve the proof system, we propose the Proof of Encrypted Storage (PoES) algorithm. PoES enables a storage miner $\mathsf{SM}$ to store encrypted data and efficiently provide PoSs and PoSts at low cost. The PoES algorithm is represented by a set of polynomial-time algorithms, denoted as $$\mathsf{PoES=(Setup, Prove, CycleProve, Verify)},$$ illustrated in Algorithm~\ref{alg:PoES}. The $\mathsf{PoES.Prove}$ algorithm is utilized to provide PoSs, whereas the $\mathsf{PoES.CycleProve}$ algorithm is intended for PoSts. Further details are provided below.

\begin{algorithm}[!htb]
	\DontPrintSemicolon
	\caption{Proof of Encrypted Storage (PoES)}
	\label{alg:PoES}
	
 	$\triangleright$ \textcolor{BrickRed}{$\mathsf{PoES.Setup}$ (by a client)}\;
	\textbf{Inputs:} $\mathsf{ctr, pk, F, \overrightarrow{\mathsf{SM}}}$\\
        \textbf{Outputs:} $\overrightarrow{\mathsf{CID}}$\\
	Calculate the $i$th increment $\mathsf{F}[i]$ for $\mathsf{F}$\\
	\textcolor{gray}{\# generate $j$ replicas for the $i$th increment}\\
	\While{$j \leq \mathsf{ctr}$}{
            $
            \mathsf{F}_E^j \leftarrow \left\{
               \begin{array}{ll}
                     \mathsf{RSA.Enc}(\sk_j,\mathsf{F}[i]), &  \text{Plan A}  \\
                     \mathsf{PRE.Enc}(\pk_j,\mathsf{F}[i]), &  \text{Plan B}
                \end{array}
            \right.
            $\\
		Generate the file identifier $\mathsf{CID}_{[i,j]}$ for $\mathsf{F}_E^j$\\
		$\mathsf{SM}_r \leftarrow \mathsf{RandomSelect}(\overrightarrow{\mathsf{SM}})$\\
		$\mathsf{DSN.Put}(\mathsf{F}_E^j, \mathsf{SM}_r)$\\
		$\overrightarrow{\mathsf{CID}}[i][j]=\mathsf{CID}_{[i,j]}$\\
	}
	
	$\triangleright$ \textcolor{BrickRed}{$\mathsf{RandSelect}$ (by a client)}\;
	\textbf{Inputs:} $\overrightarrow{\mathsf{SM}}$\\
        \textbf{Output:} $\mathsf{SM}_r$\\	
	\For{$\mathsf{SM}_{i}$ in $\overrightarrow{\mathsf{SM}}$} {
		Get its consensus power $\mathsf{POW}_{\mathsf{SM}_{i}}$\\
		$\Delta H \leftarrow \mathsf{DealTime}(\mathsf{SM}_{i})$\\
  $P_{\mathsf{SM}_{i}} \leftarrow \tilde{\Gamma}(\mathsf{POW}_{\mathsf{SM}_{i}}, \Delta H)$\\
		Add $P_{\mathsf{SM}_{i}}$ to $\overrightarrow{P}$\\
	}
	\While{$b=0$}{
		$r \leftarrow \mathsf{Gen}(1^{k})$ and $\mathsf{SM}_r \leftarrow \mathsf{Select}(r,\overrightarrow{P})$\\
		Create a deal with $\mathsf{SM}_r$ and set $b=1$ if succeed\\
	}
	
	$\triangleright$ \textcolor{BrickRed}{$\mathsf{PoES.Prove}$ (by a storage miner)}\;
	\textbf{Inputs:} $\mathsf{c, CID}$\\
        \textbf{Outputs:} $\pi_{\mathsf{POS}}$\\ 
	Get $\tilde{\mathsf{F}}_E$ from the local storage by $\mathsf{CID}$\\
	Compute a Merkle tree $\mathcal{M}$ with root $\mathsf{rt}$ for $\tilde{\mathsf{F}}_E$\\
    Compute a path $\mathsf{\tau}_c$ as a part of the proof\\
	$\pi_{\mathsf{POS}} \leftarrow \mathsf{ZK.Prove}(\mathsf{\tau}_c, \mathsf{rt}, \mathsf{c})$\\

	$\triangleright$ \textcolor{BrickRed}{$\mathsf{PoES.CycleProve}$ (by a storage miner)}\;
	\textbf{Inputs:} $\mathsf{c, t, CID}$\\
        \textbf{Outputs:} $\pi_{\mathsf{POST}}$\\
         Get $\tilde{\mathsf{F}}_E$ from the local storage by $\mathsf{CID}$\\
         Compute a Merkle tree $\mathcal{M}$ with root $\mathsf{rt}$ for $\tilde{\mathsf{F}}_E$\\
         Set $\pi_{\mathsf{POST}}=\perp$\\
        \For {$i=1...t$}{
        	$c'=H(\pi_{\mathsf{POST}}||c||i)$\\
        	$\pi_{\mathsf{POS}}=\mathsf{PoES.Prove}(c', \mathsf{CID})$\\
        	$\pi_{\mathsf{POST}}=\mathsf{ZK.Prove}(\pi_{\mathsf{POS}}, \pi_{\mathsf{POST}}, \mathsf{rt}, c, i)$\\
		}
	$\triangleright$ \textcolor{BrickRed}{$\mathsf{PoES.Verify}$ (by a smart contract)}\;
	\textbf{Inputs:} $\pi_{\mathsf{POS}}$ (or $\pi_{\mathsf{POST}}$), $rt$, $c$\\
	\textbf{Outputs:} $b$\\
	$b=0$\\
	\If {$\mathsf{ZK.Verify}$$(\pi_{\mathsf{POS}}$ (or $\pi_{\mathsf{POST}}$)$, rt, c)=\perp$}{
		Penalize the storage miner who provides the proof\\
	} \Else{
            $b=1$\\
        }
\end{algorithm} 

$\mathsf{PoES.Setup(F, ctr, pk, \overrightarrow{\mathsf{SM}})}\rightarrow \overrightarrow{\mathsf{CID}}$: 
The process of uploading a file $\mathsf{F}$ to storage miners involves several steps. First, the client checks the version of $\mathsf{F}$ and calculates the incremental changes between the current version and the previous version. We denote each incremental change as $\mathsf{F}[i]$, where $\mathsf{F}[0]$ is the initial version. The client then creates $\mathsf{ctr}$ different encrypted replicas of $\mathsf{F}[i]$ using either RSA (Plan A) or PRE (Plan B) for various data sharing purposes. Each replica is assigned a unique content identifier $\mathsf{CID}_{[i,j]}$, indicating that it is the $j$th replica of the $i$th increment. Next, the $\mathsf{RandomSelect}$ algorithm is called to randomly select a storage miner, denoted as $\mathsf{SM}_r$, from a list of storage miners $\overrightarrow{\mathsf{SM}}$. Finally, the client uploads the replica $\mathsf{F}_E^j$ to $\mathsf{SM}_r$ using the $\mathsf{DSN.Put}$ protocol. The content identifier of each replica is recorded in $\overrightarrow{\mathsf{CID}}$ for future verification and retrieval purposes.

$\mathsf{RandSelect(\overrightarrow{\mathsf{SM}})}\rightarrow \mathsf{SM}_r$:
The algorithm responsible for selecting a storage miner to store an uploaded replica is invoked by $\mathsf{PoES.Setup}$. The client initially calculates $\overrightarrow{P}$ by traversing each storage miner [line 15-19]. Within each loop, the client obtains the consensus power of a storage miner $\mathsf{SM}_i$, denoted as $\mathsf{POW}_{\mathsf{SM}_{i}}$, which describes the service quality of  $\mathsf{SM}_{i}$. Subsequently, the block height difference $\Delta H$ between the current height and the last confirmed storage deal of $\mathsf{SM}_{i}$ is calculated. Then $\tilde{\Gamma}(\mathsf{POW}_{\mathsf{SM}_{i}}, \Delta H)$ is computed, which comprises of two steps:
(1)
$W_{\mathsf{SM}_{i}}=w\frac{\mathsf{POW}_{\mathsf{SM}_{i}}}{\sum_{i=1}^{|\overrightarrow{\mathsf{SM}}|} {\mathsf{SM}_{i}}} + (1-w)\frac{\Delta H}{H}$;
and (2)
$P_{\mathsf{SM}_{i}}=\frac{W_{\mathsf{SM}_{i}}}{\sum_{i=1}^{|\overrightarrow{\mathsf{SM}}|} W_{\mathsf{SM}_{i}}}$,
where $w\in(0,1)$ is an adjusted weight parameter set by the client. 
The normalized probability $P_{\mathsf{SM}_{i}}\in(0,1)$ represents the likelihood of selecting the storage miner $\mathsf{SM}_{i}$ as a storage provider. After $P_{\mathsf{SM}_{i}}$ is successfully obtained for each $\mathsf{SM}_i$, the client generates a random number and selects a storage miner $\mathsf{SM}_r$ based on the probability distribution $\overrightarrow{P}$ [line 20-22].  This stochastic process can effectively mitigate Sybil attacks, as it is arduous for an attacker to simultaneously obtain the tasks of storing replicas of a file.

$\mathsf{PoES.Prove(c, CID)}\rightarrow \pi_{\mathsf{POS}}$: 
The proof algorithm involves a storage miner who generates a proof of storage ($\pi_{\mathsf{POS}}$) for an encrypted replica that has been stored locally. Initially, the miner retrieves the replica $\tilde{\mathsf{F}}_E$ from its local storage using $\mathsf{CID}$. Then the miner computes a Merkle tree $\mathcal{M}$ whose root is $\mathsf{rt}$ for $\tilde{\mathsf{F}}_E$. 
The process of creating the Merkle tree involves partitioning $\tilde{\mathsf{F}}_E$ into 256-byte chunks, and each chunk is treated as a leaf of $\mathcal{M}$. The storage miner then selects a specific leaf in the tree using the random challenge $c$ and determines a path $\mathsf{\tau}_c$ from that leaf to $\mathsf{rt}$. Finally, the miner utilizes $\mathsf{ZK.Prove}$ to produce a succinct proof $\pi_{\mathsf{POS}}$ for the entire process. This proof is vital in verifying the accuracy of the process [line 27-29].

$\mathsf{PoES.CycleProve(c, t, CID)}\rightarrow \pi_{\mathsf{POST}}$: 
In the cycle prove algorithm, a storage miner is required to generate a proof of spacetime ($\pi_{\mathsf{POST}}$) to demonstrate that it is continuously storing the data. This proof can be quickly verified by other miners. $\mathsf{PoES.CycleProve}$ is analogous to a multi-round $\mathsf{PoES.Prove}$. In each round, the storage miner first generates a round-challenge $c'$ using the input challenge $c$, the round number $i$, and the proof of spacetime from the previous round. Then, the storage miner calls $\mathsf{PoES.Prove}$, which takes $c'$ and $\mathsf{CID}$ as inputs, to generate a storage proof $\pi_{\mathsf{POS}}$. Finally, the storage miner calls $\mathsf{ZK.Prove}$ to generate a temporal proof for the current round $i$. After $t$ rouunds, $\mathsf{PoES.CycleProve()}$ outputs the final $\pi_{\mathsf{POST}}$.

$\mathsf{PoES.Verify(}$$\pi_{\mathsf{POS}}$ (or $\pi_{\mathsf{POST}}$) $, rt, c) \rightarrow 1/0$: 
The smart contract checks if a proof ($\pi_{\mathsf{POST}}$ or $\pi_{\mathsf{POS}}$) is a valid zk proof using the verify process of zk-SNARK. If the proof passes the verification process, the algorithm outputs 1; otherwise outputs 0 and panelizes (e.g, broadcast a message to suggest $P_{\mathsf{SM}_{i}}=0$) the storage miner who provides the proof.

\subsection{Batch Verification and File Retrieval}
\label{sec:multi:version}
To tackle the issue stated in [C3], we suggest a novel batch verification approach that relies on the concept of zk-rollup. Batch verification entails transferring the verification of many proofs to an aggregated succinct proof, which effectively reduces the computational and verification workload on the blockchain. Our proposed batch verification and file retrieval scheme comprises of a set of polynomial-time algorithms, denoted as $$\mathsf{Rollup = (Prepare, Collect, Aggregate)},$$ as well as a function $\mathsf{Retrieve}()$.

\begin{algorithm}[!t]
	\DontPrintSemicolon
	\caption{Batch Verification and File Retrieval}
    \label{alg:Mul_Verify}
	$\triangleright$ \textcolor{BrickRed}{$\mathsf{Rollup.Prepare}$ (by a storage miner)}\;
	\textbf{Inputs:} $\mathsf{\mathsf{CID}, RoM}$\\
		Send a $\mathsf{CID}$ to a rollup miner $\mathsf{RoM}$\\
		Wait for an aggregated proof confirmed on chain\\
	
	$\triangleright$ \textcolor{BrickRed}{$\mathsf{Rollup.Collect}$ (by a rollup miner)}\;
	\textbf{Inputs:} $\overrightarrow{\mathsf{CID}}, t$\\
	\textbf{Output:} $\overrightarrow{\pi_{\mathsf{POST}}}$\\
		Get the $\overrightarrow{\mathsf{SM}}$ who stores the increments implied by $\overrightarrow{\mathsf{CID}}$ \\
		\For{each $\mathsf{SM}$ (storing $\mathsf{CID}$) in $\overrightarrow{\mathsf{SM}}$}{
			Obtain a random challenge $c$\\
			Request the $\mathsf{SM}$ execute $\mathsf{PoES.CycleProve}(c, t, \mathsf{CID})$\\
			AsycWait for a reply $\pi_{\mathsf{POST}}$\\
			\If {$\pi_{\mathsf{POST}}$ is valid}{
				Add $\pi_{\mathsf{POST}}$ to $\overrightarrow{\pi_{\mathsf{POST}}}$\\
			}
		}

	$\triangleright$ \textcolor{BrickRed}{$\mathsf{Rollup.Aggregate}$ (by a rollup miner)}\;
	\textbf{Inputs: } $\overrightarrow{\pi_{\mathsf{POST}}}$\\
	\textbf{Output:} $\pi_{\mathsf{ROLL}}$\\
	Get a prepared rollup $\mathsf{circuit}$ according to $|\overrightarrow{\pi_{\mathsf{POST}}}|$ \\
	Input $\overrightarrow{\pi_{\mathsf{POST}}}$ to $\mathsf{circuit}$ and obtain $\pi_{\mathsf{ROLL}}$\\
	Submit $\pi_{\mathsf{ROLL}}$ to a smart contract\\

		$\triangleright$ \textcolor{BrickRed}{$\mathsf{Retrieve}$ (by a client)}\;
		\textbf{Inputs:} $\overrightarrow{\mathsf{CID}}$\\ 
		\textbf{Outputs:} $\mathsf{F}$\\ 
		Fetch from blockchain the $\pi_{\mathsf{ROLL}}$ corresponding to $\overrightarrow{\mathsf{CID}}$\\
		\If {$\pi_{\mathsf{ROLL}}$ is a valid rollup proof}{
			\For {each $\mathsf{CID}_{[i,j]}$ in $\overrightarrow{\mathsf{CID}}$}{
				
    $
    \mathsf{F}_E\leftarrow \left\{
        \begin{array}{llll}
            \mathsf{DSN.Get}(\mathsf{CID}_{[i,j]},\mathsf{ReM}), & \text{Plan A}\\
            \mathsf{GetReEnc}(\mathsf{CID}_{[i,j]},\mathsf{ReM}), & \text{Plan B}
        \end{array}
        \right.
    $\\
    $
        \mathsf{F}_D\leftarrow \left\{
        \begin{array}{llll}
            \mathsf{RSA.Dec}(\pk,\mathsf{F}_E), & \text{Plan A}\\
            \mathsf{PRE.Dec}(\sk,\mathsf{F}_E), & \text{Plan B}
        \end{array}
        \right.      
    $      \\
				Add $\mathsf{F}_D$ to $\vec{\mathsf{F}}[i]$\\
			}
		Recover $\mathsf{F}$ by aggregating $\vec{\mathsf{F}}$\\
		} \Else {
			$\mathsf{F}=\perp$\\
		}
\end{algorithm}

$\mathsf{Rollup.Prepare(CID, RoM)}$: 
Rollup miners have the duty of producing aggregated proofs for a number of files. A storage miner executes $\mathsf{Rollup.Prepare}$ to request rollup miners to aggregate proofs. It sends a $\mathsf{CID}$ to a rollup miner $\mathsf{RoM}$, and waits for the corresponding aggregated proof confirmed on blockchain. 

$\mathsf{Rollup.Collect}(\overrightarrow{\mathsf{CID}}, t)\rightarrow \overrightarrow{\pi_{\mathsf{POST}}}$:
Rollup miners are permitted to continuously gather $\pi_{\mathsf{POST}}$ to produce an aggregated proof. The $\mathsf{Rollup.Collect}$ algorithm takes a vector of content identifiers $\overrightarrow{\mathsf{CID}}$ and $t$ as inputs to compute $\pi_{\mathsf{POST}}$. Initially, the rollup miner identifies the storage miners that keep the file increments indicated by $\overrightarrow{\mathsf{CID}}$, one miner for each $\mathsf{CID}$. Then the rollup miner asks each storage miner for a proof using a challenge value $c$. The rollup miner receives the proofs using $\mathtt{AsyncWait}$ and maintains the valid proofs in $\overrightarrow{\pi_{\mathsf{POST}}}$.

$\mathsf{Rollup.Aggregate}(\overrightarrow{\pi_{\mathsf{POST}}})\rightarrow \pi_{\mathsf{ROLL}}$:
We employ the aggregate algorithm to calculate a succinct proof $\pi_{\mathsf{ROLL}}$ of 256 bytes in length. This proof is based on multiple valid proofs $\overrightarrow{\pi_{\mathsf{POST}}}$. To generate the $\pi_{\mathsf{ROLL}}$, we utilize a rollup circuit that varies in size depending on the size of $\overrightarrow{\pi_{\mathsf{POST}}}$. Specifically, the circuit is pre-configured with different sizes, for example, a combination of \{1KB, 4KB, 8KB, $\cdots$\}. The resulting proof ($\pi_{\mathsf{ROLL}}$) is then transmitted to a smart contract for final confirmation. A brief proof is also recorded on the blockchain to facilitate on-chain verification when the corresponding file is retrieved.

$\mathsf{Retrieve}(\overrightarrow{\mathsf{CID}})\rightarrow \mathsf{F}$:
This process is performed by a client who wants to retrieve a certain file. The client first fetches $\pi_{\mathsf{ROLL}}$ corresponding to $\overrightarrow{\mathsf{CID}}$ and verifies whether the proof is valid. Then for each file piece required to recover the complete file, the client requests it from the retrieval miners by calling $\mathsf{DSN.Get}$ or $\mathsf{F}_E\leftarrow \mathsf{GetReEnc}(\mathsf{CID}_{[i,j]},\mathsf{ReM})$, and decrypts it locally to obtain $\mathsf{F}_D$. 
All the file pieces are collected in $\vec{\mathsf{F}}[i]$. The client finally recovers the file by putting all file pieces together. 
Figure \ref{fig:workflow} illustrates the execution of the FileDES protocols.
\begin{figure}[!t] 
	\centering 
	\includegraphics[width=0.45\textwidth]{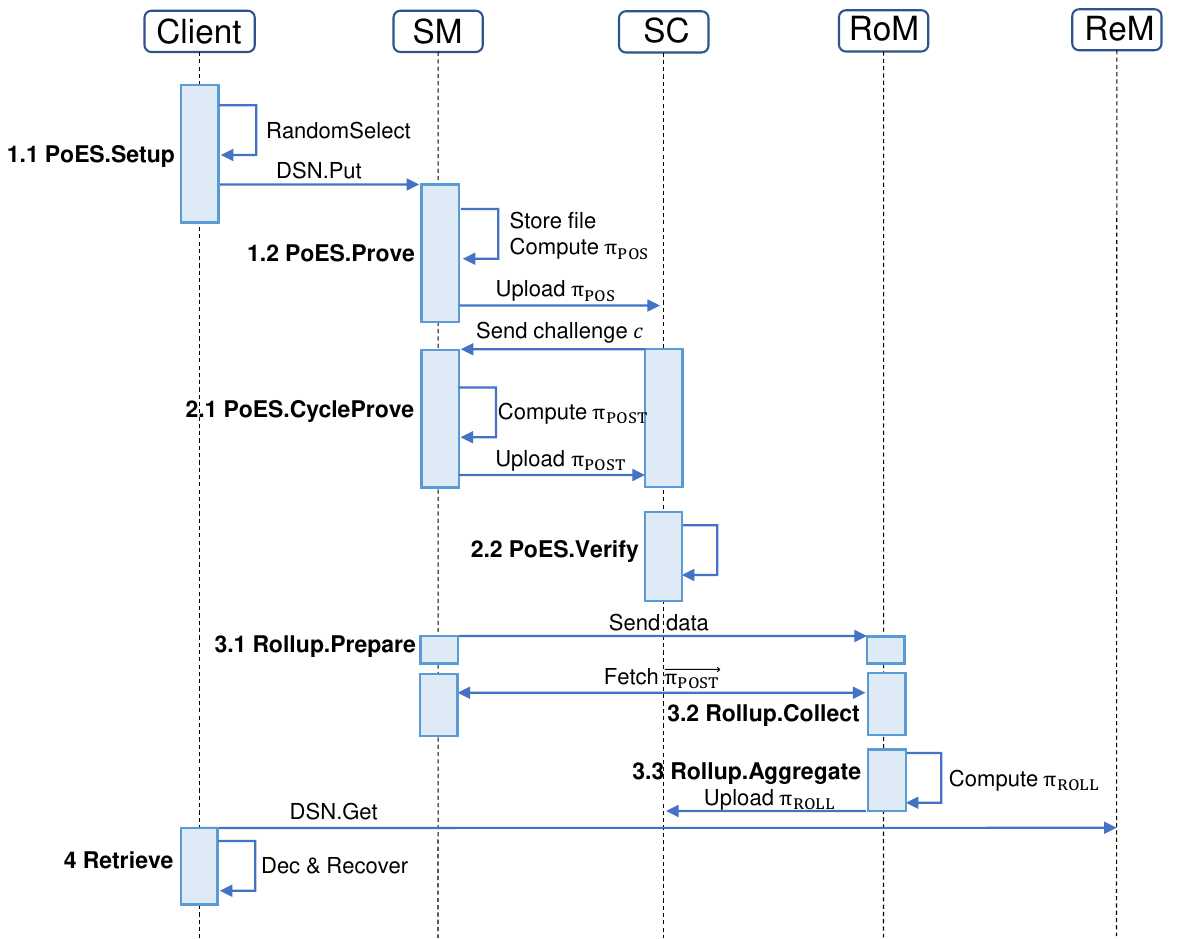} 
	\caption{An example of executing all protocols} 
	\label{fig:workflow} 
\end{figure}

\section{Security Analysis}
\label{sec:analysis}
\subsection{Security of PoES}
\begin{theorem}[Unforgeability of $\pi_{\mathsf{POS}}$/$\pi_{\mathsf{POST}}/\pi_{\mathsf{ROLL}}$]
    An honest storage miner $\mathsf{SM}$ or rollup miner $\mathsf{RoM}$ can convince storage challengers by sending them valid $\pi_{\mathsf{POS}}$/$\pi_{\mathsf{POST}}$ or $\pi_{\mathsf{ROLL}}$; an adversary $\mathcal{A}$ without honestly storing files cannot forge a valid proof based on the public information.
\label{theorem:unforgeability}
\end{theorem}
\begin{proof}
The generation of $\pi_{\mathsf{POS}}$ is carried out by $\mathsf{PoES.Prove}$, which is executed as a sub-process in creating $\pi_{\mathsf{POST}}$ through $\mathsf{PoES.CycleProve}$. Additionally, $\pi_{\mathsf{ROLL}}$ represents a zero-knowledge proof of multiple $\pi_{\mathsf{POS}}$ instances. Firstly, we demonstrate that $\pi_{\mathsf{POS}}$ is unforgeable.

Our model assumes that the adversary $\mathcal{A}$ can act arbitrarily as a Byzantine node. The unforgeability of $\pi_{\mathsf{POS}}$ can be defeated by $\mathcal{A}$ who is able to devise a $\pi_{\mathsf{POS}}'$ that can convince a challenger. A valid $\pi_{\mathsf{POS}}$ is generated through $\pi_{\mathsf{POS}} \leftarrow \mathsf{ZK.Prove}(\mathsf{\tau}_c, \mathsf{rt}, \mathsf{c})$, where $\mathsf{\tau}_c$ denotes the Merkle path acting as a private input to the circuit. The zero-knowledge property of zk-SNARK ensures that the adversary $\mathcal{A}$ cannot acquire or deduce $\mathsf{\tau}_c$ based on the public information provided by the $\mathsf{SM}$s. Moreover, the completeness of zk-SNARK guarantees that an honest $\mathsf{SM}$ with a valid $\pi_{\mathsf{POS}}$ can always convince the storage challenger $\mathsf{SC}$. In the case of a malicious $\mathsf{SM}$, the soundness of zk-SNARK makes it impossible for the $\mathsf{SM}$, with probabilistic polynomial-time (PPT) witness extractor $\mathcal{E}$, to provide a fake Merkle path (used to forge $\pi_{\mathsf{POS}}$) to deceive $\mathsf{SC}$, as depicted in Eq.~\eqref{eq:soundness}. In other words, the challenger can determine whether the storage miner provides a fake private input based on public parameters, such as the common reference string $\mathsf{crs} (\mathsf{pk,vk})$, the proof $\pi$, and the public inputs. 
	\begin{equation}  
	\label{eq:soundness}
	\begin{aligned}
	&
		\Pr{\left[
			\begin{aligned}
				& C(\mathsf{\tau}_c, \vec{w}) \neq R \\
				& {\mathrm{Verify}}(\mathsf{vk}, \mathsf{\tau}_c, \pi, \vec{w}) = 1
			\end{aligned}
			\Biggm\vert
			\begin{aligned}
				& {\mathtt{Setup}}(1^{\lambda}, C) \rightarrow (\mathsf{pk, vk})\\
				& \mathcal{A}(\mathsf{pk, vk}) \rightarrow (\mathsf{\tau}_c, \pi)\\
				& \mathcal{E}(\mathsf{pk, vk}) \rightarrow (\vec{w})
			\end{aligned}
			\right]} \\
   & \le \mathsf{negl}(\lambda)
	\end{aligned}
	\end{equation}
 where $pk$ and $vk$ are respectively the proving and verification keys of zk-SNARK.
	The zero knowledge property of zk-SNARK guarantees that the probability of a malicious PPT $\mathsf{SM}$ forging a $\pi_{\mathsf{POS}}'$ that can convince a challenger is negligible. 
 
 To create a $\pi_{\mathsf{POST}}$, a storage miner periodically executes $\mathsf{PoES.CycleProve}$, which recursively calls $\mathsf{PoES.Prove}$, and outputs a valid $\pi_{\mathsf{POST}}$. The only private input $\tau_c$ is still preserved. Similar to $\pi_{\mathsf{POS}}$, the completeness, soundness, and zero-knowledge properties guarantee that the $\pi_{\mathsf{POST}}$ cannot be forged. The $\mathsf{Rollup.Aggregate}$ bundles multiple $\pi_{\mathsf{POS}}$'s or $\pi_{\mathsf{POST}}$'s into batches and employs the blockchain to ensure unforgeability. Once the blockchain confirms a proof, it becomes tamper-proof. Furthermore, all nodes have consistent views of the blockchain ledger, which ensures that the output of $\mathsf{Rollup.Aggregate}$ is consistent and deterministic. It is worth noting that a $\pi_{\mathsf{POS}}$/$\pi_{\mathsf{POST}}$/$\pi_{\mathsf{ROLL}}$ is of short size due to the succinctness property of zk-SNARK.
\end{proof}

\begin{theorem}[Sybil and Generation Attack Resistance]
Given the assurance of security offered by the public key infrastructure and the zero-knowledge proof system, the probability of Probabilistic Polynomial-Time (PPT) adversaries $\mathcal{A}{\mathsf{sybil}}$ or $\mathcal{A}{\mathsf{gen}}$, achieving success is negligible. 
\label{theorem:sybil}
\end{theorem}

\begin{proof}
We first investigate the security of FileDES against Generation attacks. Generation attackers aim to replicate a file using a small seed or a program. This seed or program is much smaller than the actual file in size.
Recall that each leaf of a Merkle tree is a data chunk of 256 bytes. When a challenge is received, the Generation attacker $\mathcal{A}_{\mathsf{gen}}$ needs to provide the path from a leaf to the Merkle root. Since there are $O(2^h)=O(2^{\sqrt{N}})$ possible paths, where $N$ is the number of leaves, the probability of hitting a valid path is $O(\frac{1}{2^{\sqrt{N}}})$, which is negligible.
If $\mathcal{A}_{\mathsf{gen}}$ stores a constant number of paths to deceive the challenger, the probability of success remains negligible. However, if $\mathcal{A}_{\mathsf{gen}}$ stores a sufficient number of paths, e.g., covering $O(\sqrt{N})$ or $O(N)$nodes on the Merkle tree, to increase its winning probability, it will suffer significant storage overhead, which violates its original intention of saving space via a small-sized seed or program.

A Sybil attack $\mathcal{A}_{\mathsf{sybil}}$  can cheat in two ways. (1) $\mathcal{A}_{\mathsf{sybil}}$ only stores $m'<m$ ($m', m \in \mathbb{Z}$) replicas to save space. Once requested to provide a PoSt, $\mathcal{A}$ reproduces $m$ replicas and then creates proofs based on them. After accomplishing the proof procedure, $\mathcal{A}$ deletes $(m-m')$ replicas and stores only $m'<m$ replicas. (2) $\mathcal{A}_{\mathsf{sybil}}$ stores $m'<m$ replicas; and once queried, $\mathcal{A}$ forges a PoSt and cheats the verifiers. 
However, the first case is infeasible since $\mathcal{A}$ would need to decrypt the encrypted file in order to generate a proof, which would require the client's private key. Breaking the security of RSA or PRE is computationally hard under the assumption that they are secure against a PPT attacker. The success probability in the second case is also negligible, according to Theorem~\ref{theorem:unforgeability}, which prohibits $\mathcal{A}$ from forging proofs. To mitigate Sybil attacks, a random selection mechanism is introduced, in which the client selects a subset of $m$ storage miners from $n$ available ones. $\mathcal{A}$ could create fake nodes, but the probability of $m$ storage miners being selected from the true miner pool is high, which is at least $1-\frac{m}{n}$.
\end{proof}

\subsection{Consistency of FileDES}
Ensuring data consistency is a critical aspect for DSNs, as it guarantees that all storage miners have the same view on the stored files. FileDES achieves this objective, as proved by the following theorem.

\begin{theorem}[Consistency]
If an honest node proclaims a version of a file as stable, other honest nodes, if queried, either report the same result or report error messages. Here ``stable'' means that the file version is stored by FileDES and permanently recorded on the blockchain. 
\label{theorem:consistency}
\end{theorem}

\begin{proof}
A file version can be verified via the on-chain proofs stored in the blockchain,
which has been proven by Theorem~\ref{theorem:sybil} and Theorem~\ref{theorem:unforgeability} to be resilient against adversaries. Consequently, all on-chain proofs are deemed authentic. However, as the system is decentralized, inconsistencies in the proofs may occur. For example, different clients may observe various states of the same file simultaneously, resulting in significant issues when retrieving the file. To address this concern, FileDES employs DAG-Rider, a robust and efficient consensus algorithm that can tolerate byzantine faults. The DAG-Rider constitutes an atomic broadcast algorithm possessing the properties of agreement, integrity, validity, and total order. The first three properties guarantees a dependable broadcast, ensuring that all processes in a distributed system receive the same group of messages. Thus, any two nodes in the DSN attain consensus on the same set of increments and proofs. Furthermore, the total order attribute guarantees that any two increments (along with their corresponding proofs) are comparable, generating a deterministic and consistent causal history of the file updates.
\end{proof}


\section{Evaluation}
\subsection{Implementation}
The section introduces the development of FileDES, which incorporates innovative client and miner modules, along with the enhanced public service modules from Filecoin. The architectural layout of FileDES is illustrated in Figure~\ref{diagram}, in which modules shaded in gray represent the adapted ones from Filecoin, modules shaded in blue refer to those modified or improved from Filecoin, and modules shaded in green denote newly developed ones.

\begin{figure}[!htbp] 
	\centering 
	\includegraphics[width=0.45\textwidth]{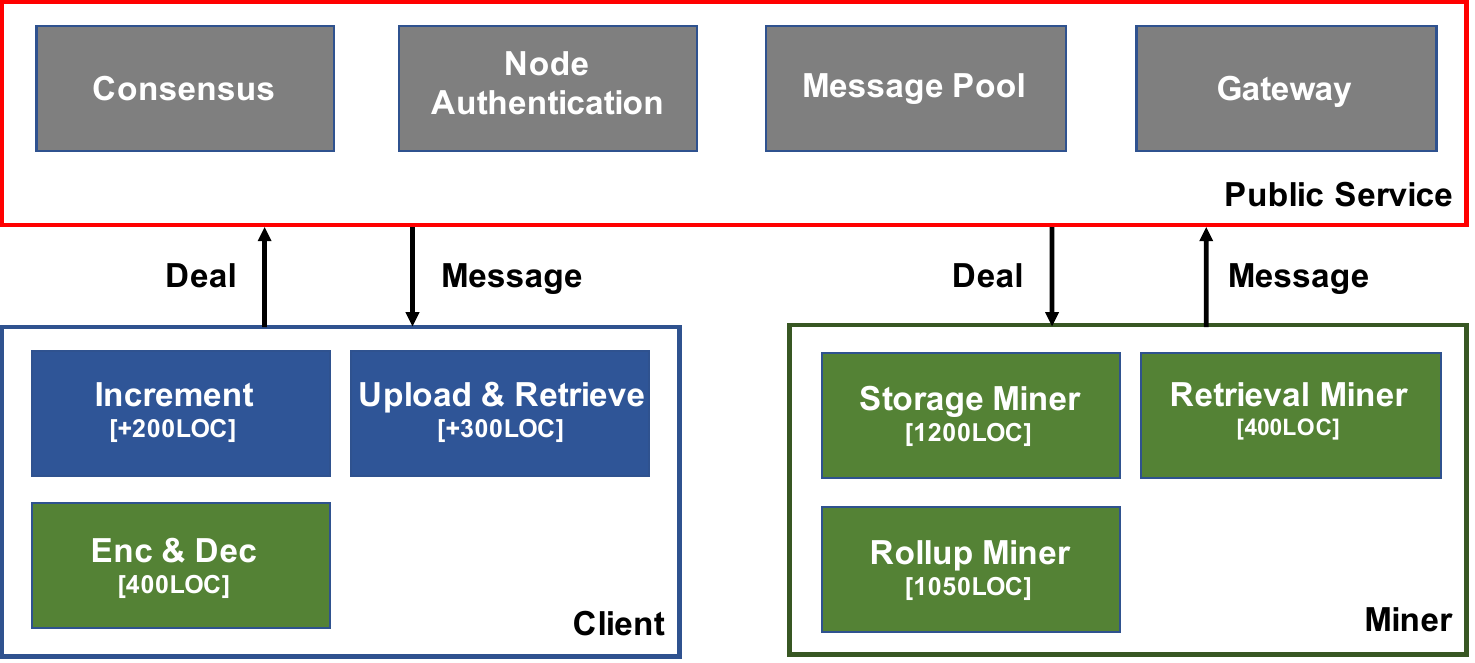} 
	\caption{The system architecture of FileDES} 
	\label{diagram}
\end{figure}

\subsection{Experiment Setup}
Our experimental study is comprised of two main segments. In the first segment, we establish a DELL PowerEdge R740 server operating on Ubuntu 22.04 LTS. The server is equipped with two 12-Core CPUs, 16GB memory, and 300GB SSD. We deploy four DSNs, namely FileDES, Filecoin, Storj, and Sia, on the server. 
To make a fair comparison, we let each DSN system prove on all sectors of a file. We exclude FileDAG and Swarm mentioned in TABLE~\ref{table:features} from our evaluation because FileDAG adopts the proof system of Filecoin while that of Swarm is not open-sourced. Our dataset consists of text files (.txt) and binary files (Android .apk) from various real-world projects, including git, go-ipfs, Minecraft, and Netflix. We observe analogous conclusions that are irrespective of the file types. Therefore, we present typical outcomes concerning text files to better elaborate on the performance of FileDES. We  also incorporate supplementary test data in the appendix for further scrutiny. This segment of study involves three tests:
\begin{itemize}
    \item the proof generation time with different file sizes;
    \item the storage cost with different size and varied number of total versions; 
    \item the proof generation and verification time with variable number of total file versions.
\end{itemize}

In the second segment, we deploy FileDES, Filecoin, Sia, and Storj in a Wide Area Network (WAN) consisting of 120 ecs.t5-lc1m4.large instances. Each instance is equipped with a 2-Core CPU, 4GB memory, and 40GB SSD, and is configured with Ubuntu 22.04 LTS. The bandwidth capacity of each instance is 100Mbps, and a single node is established on each instance. Out of the 120 instances, 100 instances are designated as storage miners while the remaining 20 instances are clients. The evaluation criteria for this phase includes the following:
\begin{itemize}
\item the throughput of proof generation as the number of clients increases; 
\item the correlation between latency and throughput.
\end{itemize}

\subsection{Performance} 

\textbf{Proof Generation Time.} 
The processes of generating PoS and PoSt were examined across various file sizes. As depicted in Figure~\ref{PoS generation}, the PoS generation times in FileDES, Sia, and Storj increase with file size, while Filecoin exhibits stability. This can be attributed to the changeable sector size in FileDES, where files are padded to different sizes to create a balanced Merkle tree. The PoS generation times in Sia and Storj increase linearly due to the increasing number of proofs required with the increase of the file size. Filecoin and Swarm, on the other hand, have a fixed sector size of 8MB, which necessitates padding files with random data to obtain the required size. Our results indicate that the proof generation time for FileDES is shorter than those of Filecoin and Storj, and close to that of Sia. However, Sia's security strength is weaker compared to that of FileDES due to its implementation of the Merkle tree with 64-byte leaves, which is larger than the one used in FileDES. This results in a reduced number of leaves generated by Sia. Our analysis, as presented in Section~\ref{sec:analysis}, reveals that the security level is determined by $O(2^{\sqrt{N}})$, with $N$ representing the number of leaves.
Figure~\ref{PoSt generation} shows that the PoSt generation time in FileDES is shorter than those in Filecoin and Storj, and close to that of Sia. Since PoSt in Filecoin does not involve a sealing process, its generation is faster than PoS. The latencies of FileDES and Sia are close because their PoSt processes are similar in generating the Merkle paths on the already-processed file sectors. 

To evaluate the encryption and decryption processes, the two encryption options were tested. RSA-based encryption takes approximately 7.5 seconds to encrypt 1MB of data, while PRE-based encryption takes approximately 5.8 seconds. RSA-based encryption takes 0.35 seconds to decrypt 1MB of data, while PRE-based encryption takes about 3.7 seconds for decryption. The time cost for PRE-based encryption to generate a re-encryption key is approximately 0.11 seconds. Although the encryption process takes time, it is still reasonable since it eases the computational burden on the storage miners by allocating the task to end-users. Consequently, the performance of the proof system is enhanced compared to that of the conventional methods.

\begin{figure}[!t]
\centering
\subfigure[PoS]{
\label{PoS generation}
\includegraphics[width=0.48\linewidth]{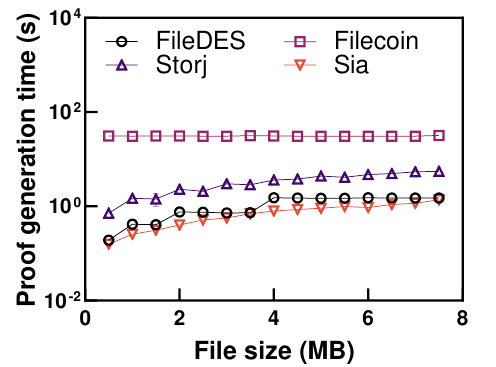}}
\subfigure[PoSt]{
\label{PoSt generation}
\includegraphics[width=0.48\linewidth]{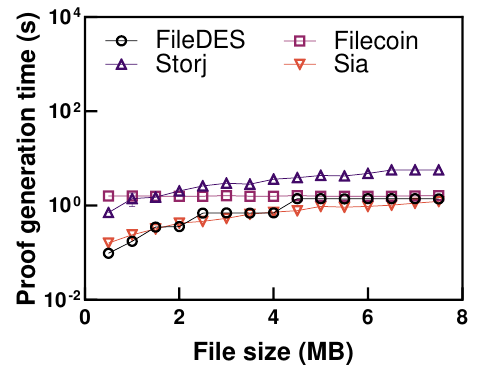}}
\caption{The proof generation time of PoS and PoSt}
\label{storage proof}
\end{figure}

\begin{figure}[!t]
\centering
\subfigure[Files of various sizes]{
\label{storage cost1}
\includegraphics[width=0.48\linewidth]{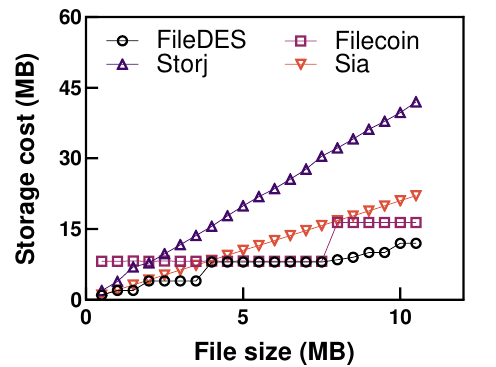}}
\subfigure[Multiple versions of a file]{
\label{storage cost2}
\includegraphics[width=0.48\linewidth]{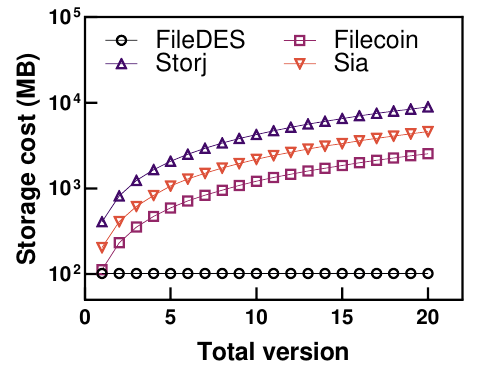}}
\caption{The storage cost of storage miners}
\label{storage cost}
\end{figure}

\textbf{Storage Cost.} 
An assessment on the storage cost, or real disk usage, for storing files of varying sizes was conducted across four systems. The results of the evaluation, presented in Figure~\ref{storage cost1}, indicate that FileDES has the lowest storage cost among the four DSNs. Interestingly, the storage cost of Filecoin is the same as that of FileDES when the file sizes are in the range of [4,7.5] MB as both systems padded files to 8MB. Storage costs in Sia and Storj increase linearly with the file size due to the use of erasure code to add redundancy to a file, resulting in an actual size of about 2.2 and 4 times of the original file size, respectively. 
Furthermore, the results shown in Figure~\ref{storage cost2} reveal that the storage cost in FileDES undergoes only a minor increase with the increasing number of versions due to the storage of only file increments whenever a multi-version file is updated to a new version. 

\textbf{PoSt Generation and Verification (Multi-version Files).} Based on Figure~\ref{multi gen2}, it is evident that FileDES has the fastest PoSt generation time compared to Filecoin, Sia, and Storj. The efficiency of FileDES can be attributed to its fast proof generation process and optimized storage of file increments. The $\mathsf{Rollup}$ function is responsible for consolidating the PoSts of each version into a single proof of constant size using a zk-SNARK circuit. However, it is crucial to limit the size of the zk-SNARK circuit to avoid overburdening the memory usage and CPU with small inputs. To tackle this problem, a limit on the number of increments used to recover a file can be set, beyond which a new base is created, thereby the maximum number of proofs to be aggregated is restricted. Figure~\ref{multi verify} depicts the total PoSt proof verification time, which remains constant at 4.5ms for FileDES as we create a succinct proof to aggregate PoSts, requiring only the verification of a single succinct proof. However, the total verification times of Filecoin, Sia, and Storj increase linearly with the total number of versions.

\begin{figure}[!t]
\centering
\subfigure[PoSt generation]{
\label{multi gen2}
\includegraphics[width=0.48\linewidth]{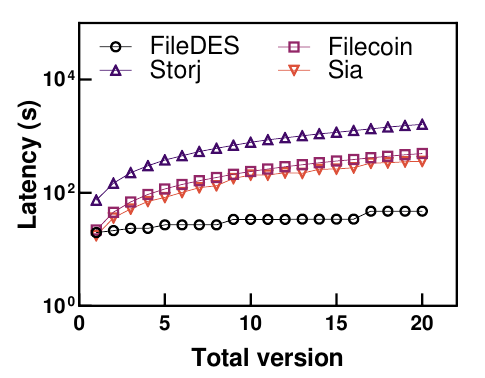}}
\subfigure[PoSt verification]{
\label{multi verify}
\includegraphics[width=0.48\linewidth]{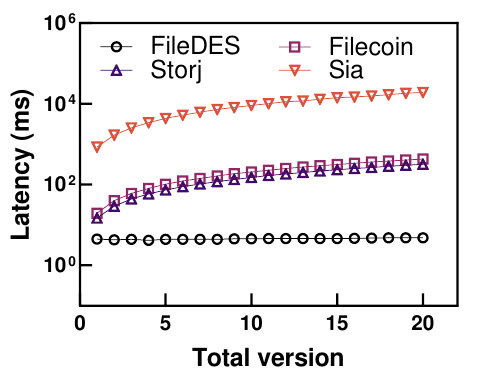}}
\caption{The latency of generating and verifying PoSts for multi-version files}
\label{multi-version file proof}
\end{figure}

\textbf{Throughput and Latency in WAN.} As far as our understanding goes, this study is the initial attempt to carry out a thorough comparative examination on the most advanced DSNs in an actual WAN environment. Particularly, this evaluation was conducted to obtain insights into the current state-of-the-art DSNs. In our experimental study, we manipulated the number of clients to send requests at random intervals, where each client dispatched 20 requests every 5 seconds. The network was consisted of 100 storage miners and up to 20 clients. The size of each uploaded file was fixed at 5MB. Our primary objective was to compare the throughputs of PoS and PoSt generation in four different Decentralized Storage Networks. Our results indicate that FileDES outperforms the other three in terms of PoS and PoSt throughputs (refer to Figure~\ref{multi throughtput}). Specifically, the throughput of FileDES is 3.02 and 1.79 times higher than that of  Filecoin in PoS and PoSt, respectively. Hence, one can infer that FileDES exhibits better scalability than the other three DSNs. The latency-throughput graph of the four DSNs is depicted in Figure~\ref{multi latency}, which reveals that FileDES consistently achieves superior performance compared to the other three DSNs under various settings. Furthermore, the latency of FileDES increases only slightly with the throughput for both PoS and PoSt generations.

\begin{figure}[!t]
\centering
\subfigure[PoS]{
\label{PoS throughput}
\includegraphics[width=0.48\linewidth]{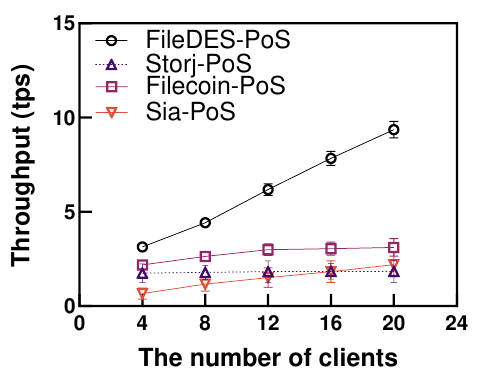}}
\subfigure[PoSt]{
\label{PoSt throughput}
\includegraphics[width=0.48\linewidth]{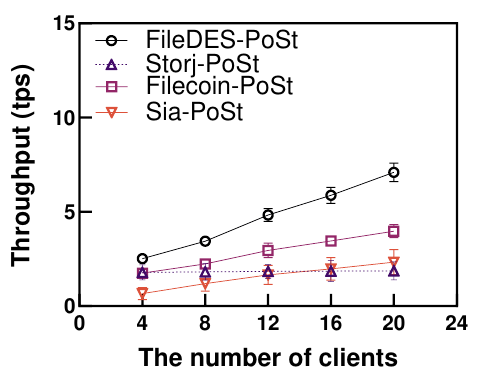}}
\caption{The throughput of PoS and PoSt with the number of clients.}
\label{multi throughtput}
\end{figure}

\begin{figure}[!t]
\centering
\subfigure[PoS]{
\label{PoS latency}
\includegraphics[width=0.48\linewidth]{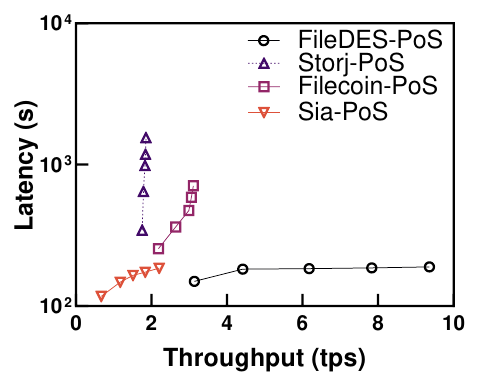}}
\subfigure[PoSt]{
\label{PoSt latency}
\includegraphics[width=0.48\linewidth]{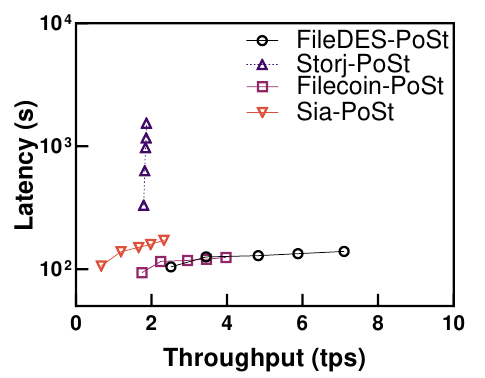}}
\caption{The latency-throughput of PoS and PoSt.}
\label{multi latency}
\end{figure}

\section{Conclusion}
\label{sec:conclusion}
This study introduces FileDES, a novel protocol that integrates three key elements: privacy preservation, scalable storage proof, and batch proof verification, for decentralized storage. The proposed protocol aims to address the exiting challenges faced by pioneers of DSN, such as data privacy leakage, costly storage proof, and low efficiency of recurrent proof verification. FileDES outperforms the state-of-the-arts in several aspects, including the proof generation/verification efficiency, storage cost, and scalability. 

\section{Acknowledgement}
This study was partially supported by the National Key R\&D Program of China (No.2022YFB4501000), the National Natural Science Foundation of China (No.62232010, 62302266), Shandong Science Fund for Excellent Young Scholars (No.2023HWYQ-008), Shandong Science Fund for Key Fundamental Research Project (ZR2022ZD02), and the Fundamental Research Funds for the Central Universities.



\bibliographystyle{IEEEtran}
\bibliography{references}

\end{document}